%% file: main.tex
\documentclass[11pt,letterpaper]{article}
\input{header}

\usepackage{enumitem}

\usepackage{authblk}

\title{Directed Global Minimum Cut in Almost-Linear Time}

\author[1]{Henry Fleischmann}
\author[1]{Jason Li}
\author[2]{Thatchaphol Saranurak}
\author[2]{Benyu Wang}
\affil[1]{Carnegie Mellon University $\{\texttt{hfleisch}, \texttt{jmli}\}$\texttt{@cs.cmu.edu}}
\affil[2]{University of Michigan $\{\texttt{thsa}, \texttt{benyuw}\}$\texttt{@umich.edu}}

\date{}
\begin{document}

\maketitle

\begin{abstract}
We give a randomized reduction from global minimum cut in a \emph{directed} graph
with $n$ vertices and $m$ weighted edges to maximum-flow computations
on graphs of total size $\widetilde{O}(m)$. For polynomially bounded
integral weights, this yields the first almost-linear $m^{1+o(1)}$-time
algorithm~\cite{vandenbrand2023deterministic}, improving the previous
$m^{1+o(1)}\min\{\sqrt{n},n/m^{1/3}\}$ bound~\cite{cen2022minimum}.

Previous work reduced finding the minimum cut to constructing  a \emph{1-respecting arborescence}, a directed
spanning tree with exactly one edge crossing a minimum cut \cite{cen2022minimum}. They then construct such a tree by sampling from a $(1+\eps)$-approximate arborescence packing, which currently requires superlinear time. Our algorithm sidesteps constructing the packing entirely and relies on a novel iterative arborescence sampling procedure. 
We show that, after $O(\log n)$ rounds of sampling, our final arborescence 1-respects the minimum cut with constant probability.
\end{abstract}

\paragraph{AI disclosure.}
The main result of this paper was discovered by 
OpenAI's GPT-6 Astra, 
with the suggestion of using length-constrained expander hierarchies to give a faster algorithm for $(1+\eps)$-approximate arborescence packing. 
The authors have made an effort to simplify and improve the exposition.
The authors assume responsibility for all content. 

\input{intro}

\input{tree-sampling-to-mincut}

\input{implementing-tree-sampling}

\input{flow-routine}


\section*{Acknowledgments}
This material is based upon work supported by the National Science Foundation Graduate
Research Fellowship Program under Grant No(s).\ DGE2140739 and DGE2631988. Any opinions, findings,
and conclusions or recommendations expressed in this material are those of the author(s)
and do not necessarily reflect the views of the National Science Foundation.

Saranurak is supported by NSF Grant CCF-2238138 and a Sloan Fellowship.
{\small
\bibliographystyle{alpha}
\bibliography{references}
}

\newpage
\appendix

\end{document}

%% file: header.tex
\usepackage[T1]{fontenc}
\usepackage{babel}
\usepackage{verbatim}
\usepackage{bm}
\usepackage{amsmath,amstext,amsthm,amssymb}
\usepackage{thmtools}
\usepackage{thm-restate}
\usepackage{standalone}
\usepackage{float}
\usepackage{xspace}
\usepackage{graphicx}
\usepackage{mathtools}
\usepackage{subcaption,array}
\usepackage{wrapfig}
\usepackage{booktabs}
\usepackage{makecell}
\usepackage{array}
\usepackage{multirow}
\usepackage[usenames,dvipsnames]{xcolor}
\usepackage{tikz}
\usetikzlibrary{arrows.meta}
\usepackage{wrapfig}
\usepackage{boxedminipage,wrapfig}
\usepackage{pifont}%

\usepackage[letterpaper,top=1in,bottom=2cm,left=1in,right=1in,marginparwidth=1.75cm,footskip=1cm]{geometry}
\usepackage[font={small,it}]{caption}
\usepackage{amssymb}
\usepackage{amsthm}
\usepackage{comment}
\usepackage{mdframed}
\usepackage{xspace}
\usepackage{multicol}
\usepackage{tabularx}
\usepackage{makecell}
\usepackage{parskip}

\usepackage[ruled,vlined,linesnumbered,algonl,procnumbered]{algorithm2e}
\SetEndCharOfAlgoLine{}
\SetKwComment{Comment}{\footnotesize$\triangleright$\ }{}
\SetCommentSty{mycommfont}

\usepackage{enumitem}
\usepackage{nicefrac}

\makeatletter
\@ifundefined{newcounteralias}{}{%
  \renewcommand\thmt@autorefsetup{%
    \@xa\def\csname\thmt@envname autorefname\@xa\endcsname\@xa{\thmt@thmname}%
  }%
}
\makeatother

\newtheorem{theorem}{Theorem}[section]
\newtheorem{lemma}[theorem]{Lemma}

\newtheorem{claim}[theorem]{Claim}
\newtheorem*{claim*}{Claim}

\allowdisplaybreaks

\declaretheoremstyle[
    bodyfont=\normalfont  %
]{upright}

\usepackage{graphicx}
\usepackage[colorlinks=true, allcolors=blue]{hyperref}
\usepackage[nameinlink]{cleveref}
\newcommand{\eps}{\varepsilon}

\newcommand{\calA}{\mathcal{A}}

\newcommand{\bw}{\textbf{w}}
\newcommand{\bc}{\textbf{c}}
\newcommand{\bx}{\textbf{x}}
\newcommand{\bd}{\textbf{d}}
\newcommand{\bg}{\textbf{g}}
\newcommand{\bh}{\textbf{h}}
\newcommand{\bq}{\textbf{q}}

\DeclareMathOperator{\poly}{poly}

\newcommand{\EE}{\mathbb{E}}

\newcommand{\RR}{\mathbb{R}}

\newcommand{\ZZ}{\mathbb{Z}}

\usepackage[colorinlistoftodos,textsize=tiny,textwidth=2cm,color=green!50!gray]{todonotes}

\newlength\hwhatcwidth%
\newlength\hwhatuwidth%
\newcommand*\halfwidehat[1]{%
  \settowidth\hwhatcwidth{\ensuremath{#1}}%
  \settowidth\hwhatuwidth{\ensuremath{\mathrm{#1}}}%
  \addtolength\hwhatuwidth{-\hwhatcwidth}%
  \makebox[0pt][l]{%
    \kern\dimexpr0.25\hwhatcwidth-0.667\hwhatuwidth\relax%
    \ensuremath{\widehat{\vphantom{#1}\rule{0.5\hwhatcwidth}{0pt}}}%
  }#1%
}%

\usepackage{todonotes}

\newcommand{\eat}[1]{}

\crefname{observation}{observation}{observations}
\Crefname{observation}{Observation}{Observations}
\Crefname{proposition}{Proposition}{Propositions}
\Crefname{open}{Open Question}{Open Questions}

\Crefname{algocf}{Algorithm}{Algorithms}
\crefname{algocfline}{line}{lines}
\Crefname{invariant}{Invariant}{Invariants}
\Crefname{claim}{Claim}{Claims}
\Crefname{subclaim}{Subclaim}{Subclaims}
\Crefname{thm}{Theorem}{Theorems}
\Crefname{lem}{Lemma}{Lemmas}
\Crefname{restatelem}{Lemma}{Lemmas}

%% file: intro.tex
\section{Introduction}

The directed minimum cut problem asks for a minimum-weight set of edges whose deletion forbids some vertex $s$ from reaching another vertex $t$. It is a fundamental robustness measure of directed graphs that has been studied for five decades.

Early algorithms for directed edge connectivity include Even--Tarjan~\cite{even1975network} and Mansour--Schieber~\cite{mansour1989finding}. Hao--Orlin~\cite{hao1994faster} solve weighted directed mincut in $\widetilde O(mn)$ time. Gabow's matroid approach~\cite{gabow1995matroid} gives $\widetilde O(m\lambda)$ time for unweighted graphs, where $\lambda$ is the minimum cut value. Chekuri--Quanrud~\cite{chekuri2021rooted} obtain near-quadratic time for unweighted simple digraphs. For general weighted graphs, the previous fastest algorithm, due to Cen et al.\ \cite{cen2022minimum}, runs in $m^{1+o(1)}\min\{\sqrt n,n/m^{1/3}\}$ time. 
Recent work obtains almost-linear time approximations~\cite{jiang2026approximating,quanrud2026approximating,mosenzon2026almost}. However, almost-linear-time exact algorithms have remained open until our work.

We show the first almost-linear-time directed minimum cut algorithm by giving a near-linear-size reduction to maximum flow. Throughout, we assume that the input directed graph $G=(V,E)$ has $n$ vertices and $m\geq n$ edges with positive \emph{integral} weights $\{w_e\}_{e\in E}$ at most $W$ and $m,W=\poly(n)$; $\widetilde O$ hides polylogarithmic factors.
\begin{theorem}\label{thm:main-result}
Directed minimum cut can be computed with high probability using maximum-flow computations on graphs of total size $\widetilde O(m)$, with $\widetilde O(m)$ additional time.
\end{theorem}
Using almost-linear time maximum flow~\cite{vandenbrand2023deterministic}, the resulting algorithm runs in $m^{1+o(1)}$ time.

As applications, we immediately obtain the first almost-linear-time algorithms for computing a second $(s,t)$-mincut and the \emph{value} of arboricity. Baswana, Bhanja, and Roy~\cite{baswana2025second} reduce second $(s,t)$-mincut---the smallest $(s,t)$-cut whose capacity is strictly greater than the $(s,t)$-mincut---to maximum flow and directed mincut, with linear total instance size. Cen et al.\ \cite{cen2025arboricity} reduce computing the arboricity value---the minimum number of forests needed to cover all edges of an undirected graph---to logarithmically many maximum-flow and directed-mincut computations on graphs of linear size.

Our key technique is an algorithm for updating a random arborescence until it crosses the mincut only once, with constant probability. This is sufficient to obtain a mincut by a known reduction of \cite{cen2022minimum}.
\Cref{sec:tree-sampling} states this key technique and proves \Cref{thm:main-result}. \Cref{sec:implementing-sampling} shows how to implement the technique assuming one flow subroutine, which is then proved in \Cref{sec:flow-sub}.

%% file: tree-sampling-to-mincut.tex
\section{Mincut via Iterative Arborescence Sampling}\label{sec:tree-sampling}

\paragraph{Rooted minimum cut.}
Write $\delta^-(S)=E(V\setminus S,S)$ and $\bw(F)=\sum_{e\in F}w_e$. A \emph{rooted cut} is a nonempty set $S\subseteq V\setminus\{r\}$. The weight of \emph{rooted mincut} is denoted by $\lambda=\min_S\bw(\delta^-(S))$. Directed global mincut is the smaller of the rooted mincuts in $G$ and its edge reversal, for any fixed $r$: a global optimum either excludes $r$, or its complement is a rooted cut in the reversed graph. We therefore focus on rooted mincut.

\paragraph{Sufficient goal: $1$-respecting arborescence.}
An \emph{$r$-arborescence} is a spanning tree whose edges point away from $r$. It \emph{$1$-respects} a rooted cut $S$ if $|T\cap\delta^-(S)|=1$. Every arborescence crosses each rooted cut at least once. The following reduces the problem to finding a $1$-respecting arborescence.

\begin{theorem}[Cut recovery, {\cite[Theorem 2.7]{cen2022minimum}}]\label{thm:1-respect-reduct}
Given an $r$-arborescence $T$, a valid rooted cut can be found using $O(\log n)$ maximum-flow computations on $O(m)$-size graphs, with $\widetilde O(m)$ additional time. If $T$ $1$-respects a rooted mincut, the returned cut is a rooted mincut.
\end{theorem}

\paragraph{Previous approach: Sampling once from arborescence packing.}
An $\alpha$-approximate fractional arborescence packing where $\alpha<2$  suffices: sampling a tree from the packing gives expected crossings at most $\alpha$, hence probability at least $2-\alpha$ of being $1$-respecting. 
Cen et al.\ \cite{cen2022minimum} compute $(1+\eps)$-approximate packings, but their direct implementation takes $\widetilde O(m\lambda)$ time on unweighted graphs, which is too slow for us. Jiang et al.\ \cite{jiang2026approximating} obtain an almost-linear time packing but with $n^{o(1)}$ congestion, insufficient for this purpose. An almost-linear time algorithm for exact or $(1+\eps)$-approximate fractional arborescence packing remains an interesting open problem.

\paragraph{Our approach: Iterative arborescence sampling.}
We instead improve the marginals of one random tree. Let
$\calA_r=\operatorname{conv}\{\mathbf1_T:T\text{ is an }r\text{-arborescence}\}$.
We say vector $\bq$ \emph{dominates} vector $\bx$ if $\bx \leq \bq$ coordinatewise; all vector inequalities below are similarly coordinatewise. 
\begin{restatable}[Marginal Contraction]{restatelem}{marginalcontraction}\label{lem:marginal-contraction}\label{lem:mixed-tree-sampling}
Given an $r$-arborescence $T$ and a nonnegative vector $\bq$ dominating some point of $\calA_r$, we can sample an $r$-arborescence $T'$ with
\[
  \EE[\mathbf1_{T'}\mid T]\leq\frac{\mathbf1_T+\bq}{2}
\]
using one maximum-flow computation on a $\widetilde O(m)$-size graph and $\widetilde O(m)$ additional time.
\end{restatable}
We emphasize that the feasible point $\bx \in \calA_r$ dominated by $\bq$ is neither given nor computed. Each update halves the excess of the inclusion probabilities above $\bq$. We prove the lemma in \Cref{sec:implementing-sampling,sec:flow-sub}.

\begin{proof}[Proof of \Cref{thm:main-result}]
Use Quanrud's approximation algorithm~\cite[Theorem 3.4]{quanrud2026approximating} to find a rooted cut of weight $C$ with $\lambda\leq C\leq5\lambda/4$, with high probability.\footnote{We can avoid calling Quanrud's algorithm by using standard binary search to find $C$ such that $\lambda\leq C\leq 5\lambda/4$.
}
It uses $O(\log^4 n)$ maximum-flow computations on $O(m)$-size graphs and $\widetilde O(m)$ additional time.
Set $L=4C/5$, so $L\leq\lambda<4L/3$. Now, define $\bq$ such that, for each edge $e$, $$\bq(e)=w_e/L.$$ By Edmonds' packing theorem~\cite{edmonds1973edge}, viewing weights as edge multiplicities gives $\lambda$ edge-disjoint arborescences. Their average $\bx\in\calA_r$ satisfies $x_e\leq w_e/\lambda\leq \bq(e)$.

Starting from any arborescence $T_0$, apply \Cref{lem:marginal-contraction} for $t$ rounds with fresh randomness. This gives
\[
  \mathbf p_t:=\EE[\mathbf1_{T_t}]
  \leq(1-2^{-t})\bq+2^{-t}\mathbf1_{T_0}.
\]
Fix a rooted minimum cut $S^*$, let $t=\lceil\log_2(16n)\rceil$, and write $X=|T_t\cap\delta^-(S^*)|$. Then
\[
  \EE[X]\leq\bq(\delta^-(S^*))+2^{-t}(n-1)
  =\frac{\lambda}{L}+2^{-t}(n-1)<\frac32.
\]
Since $X\geq1$ is an integer, $\EE[X]\geq2-\Pr[X=1]$, so $T_t$ $1$-respects $S^*$ with probability greater than $1/2$.

Repeat the above independently $O(\log n)$ times. Recover a cut from every final tree using \Cref{thm:1-respect-reduct} and return the minimum candidate. With high probability, one tree $1$-respects $S^*$, so the returned cut is a rooted mincut. 

Sampling and recovery use $O(\log^2 n)$ maximum-flow computations on $\widetilde O(m)$-size graphs and $\widetilde O(m)$ additional time. Together with Quanrud's algorithm, the aggregate size and additional time are $\widetilde O(m)$. 
\end{proof}

%% file: implementing-tree-sampling.tex
\section{Implementing the Marginal Contraction}
\label{sec:implementing-sampling}
We now prove \Cref{lem:marginal-contraction}, assuming one flow subroutine. 
\begin{restatable}[Support Mixing Flow]{restatelem}{newflowsub} \label{lem:new-flow-subroutine}
Let $T \subseteq G$ be an $r$-arborescence. Let $\bq: E \to \mathbb{R}_{\geq 0}$ such that there exists $\bx\in\calA_r$ where $\bx\leq\bq$.
There exists a randomized algorithm constructing an \emph{integral} flow $\mathbf{f} : E \to \mathbb{Z}_{\geq 0}$ such that:
\begin{enumerate}[label = (\alph*)]
    \item $\EE[\mathbf{f}] \leq \bq$, \label{cond:marg-pres}
    \item $\mathbf{f}(E) = n-1$, and \label{cond:edge-num}
    \item for all non-empty $S \subseteq V$ with $r \not\in S$, 
$(\mathbf{1}_T + \mathbf{f})(\delta^{-}(S)) \geq 2$. \label{cond:cut-lb}
\end{enumerate}
The algorithm uses one max flow call to a graph of size $\widetilde{O}(m)$, plus $\widetilde{O}(m)$ additional time.
\end{restatable}
We defer the proof of \Cref{lem:new-flow-subroutine}
to \Cref{sec:flow-sub}. Given \Cref{lem:new-flow-subroutine}, we can deduce~\Cref{lem:mixed-tree-sampling}. Let $G'$ be the multigraph with edge multiplicities $\mathbf{1}_T+\mathbf{f}$. By \Cref{lem:new-flow-subroutine}\ref{cond:cut-lb}, for each rooted cut $S \subseteq V$ (with $r \not\in S$), $G'$ has at least $2$ edges crossing into $S$. Then, by Edmonds' packing theorem \cite{edmonds1973edge}, $G'$ contains two edge-disjoint $r$-arborescences. Since \Cref{lem:new-flow-subroutine}\ref{cond:edge-num} implies that $G'$ is a graph on exactly $2n-2$ multi-edges, $G'$ is in fact a union of two edge disjoint $r$-arborescences. There are several known algorithms to efficiently compute this decomposition, e.g., Tarjan's algorithm:
\begin{theorem}[\cite{tarjan1976edge}] \label{thm:2-arbor-decomp}
Let $H = (V,E)$ be a directed graph and $s \in V$ fixed. There exists an $O(|E| \log |E|)$ time algorithm\footnote{Tarjan's algorithm actually has even better running-time, but we upper bound it by $O(|E| \log |E|)$ for simplicity.} for computing two edge-disjoint $s$-arborescences in $H$ if they exist.
\end{theorem}

The algorithm for \Cref{lem:mixed-tree-sampling} is then simple; see \Cref{alg:tree-sampler}.
\begin{algorithm}[h]
\LinesNotNumbered
\caption{Tree Sampling Algorithm}
\label{alg:tree-sampler}

\begin{enumerate}[leftmargin=*]
        \item \label{step:flow}
        Apply \Cref{lem:new-flow-subroutine} to compute $\textbf{f}$.
        \item \label{step:decomp}
        Let $G'$ be the multigraph with edge multiplicities
        $\mathbf{1}_T+\mathbf{f}$.
        Apply \Cref{thm:2-arbor-decomp} to partition its edges
        into two edge-disjoint $r$-arborescences, $T_1$ and $T_2$.
        \item Choose $i\in\{1,2\}$ uniformly at random and return 
        $T_i$.
        \end{enumerate}
\end{algorithm}

Let $T'$ be the arborescence outputted by \Cref{alg:tree-sampler}. To deduce \Cref{lem:mixed-tree-sampling}, it remains to show that 
\[
\EE[\mathbf{1}_{T'}] \leq \frac{\mathbf{1}_T + \bq}{2}.
\]
For this, note that, since we sample $T_1$ or $T_2$ uniformly at random, we have
\begin{align*}
\EE[\mathbf{1}_{T'}] &= \frac{1}{2} \EE[\mathbf{1}_{T_1} + \mathbf{1}_{T_2}] \\
&= \frac{1}{2} \left(\mathbf{1}_{T} + \EE[\mathbf{f}]\right) \\
&\leq \frac{\mathbf{1}_{T} + \bq}{2}.
\end{align*}
The last inequality uses \Cref{lem:new-flow-subroutine}\ref{cond:marg-pres}.

%% file: flow-routine.tex
\section{Support Mixing Flow Subroutine} \label{sec:flow-sub}
In this section, we prove \Cref{lem:new-flow-subroutine}, the last remaining detail in our argument. We recall the statement below.
\newflowsub*
The proof goes as follows: We first construct a flow network $H_T'$ that gives a \emph{fractional} flow $\mathbf{h} \leq \mathbf{q}$ satisfying \Cref{lem:new-flow-subroutine}\ref{cond:edge-num} and \Cref{lem:new-flow-subroutine}\ref{cond:cut-lb}. We then perform flow rounding as in \Cref{thm:flow-rounding} and show the rounded flow will satisfy all conditions in \Cref{lem:new-flow-subroutine}. 

In order to give intuition for the flow network construction, we first describe a simple flow network $H_T$ which has all the desired properties, except it is too large to be used efficiently. We then describe how to emulate this flow network on a much smaller graph $H_T'$ using heavy-light decompositions and segment trees.

\subsection{Flow Network Construction: Warmup}

The flow network $H_T$ uses three types of vertices:
\begin{itemize}
    \item a ``source'' vertex $s_v$ for each original vertex $v \in V \setminus \{r\}$, 
    \item a vertex $y_e$ for each original edge $e \in E$, and
    \item a ``sink'' vertex $t_f$ for each tree edge $f \in E(T)$.
\end{itemize}
For each edge $e = (u,v) \in E$, let $P_T(e)$ be the set of tree edges on the unique \emph{undirected} $u$-$v$ path in $T$. This is the rest of the fundamental cycle formed by adding $e$ to $T$. 

The graph $H_T$ also has the following edges for each $(u,v) = e \in E$:
\begin{itemize}
    \item a corresponding edge $a_e = (s_v, y_e)$ of capacity $\bq(e)$, and
    \item for each $f \in P_T(e)$, an infinite capacity auxiliary edge $a_{(e,f)} = (y_e, t_f)$. 
\end{itemize}
Finally, we define our vertex demand as follows:
\begin{itemize}
    \item for each non-root node $v \in V \setminus \{r\}$, we set $\bd(s_v) = -1$;
    \item for each $f\in E(T)$, we set $\bd(t_{f}) = 1$.
\end{itemize}
All other nodes have demand $0$. Negative demand means that the vertex ``demands'' to send that much net flow out; conversely, positive demand means the vertex ``demands'' to receive that much net flow in. See \Cref{fig:flow_network}.
\begin{figure}
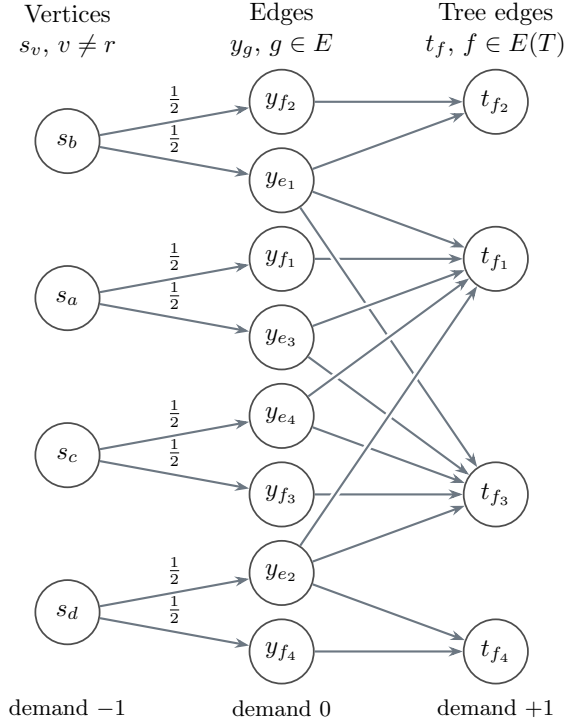

\centering
\includestandalone[mode=tex,width=0.8\textwidth]{flow-network}
\caption{Illustration of the flow network $H_T$.}
\label{fig:flow_network}
\end{figure}

The flow network has a nice interpretation in terms of $r$-arborescences. We define $\bh := \bg_0|_E$ as a flow in the original graph where $\bh(e) = \bg_0(a_e)$. The flow $\bh$ can be viewed as a sort of fractional arborescence in $G$. Indeed, without the capacity constraints on edges $a_e$, for each $(u,v) = f \in E(T)$, we could just route $1$ unit of flow from $s_v$ to $y_f$ to $t_f$. The flow $\bg_0$ corresponds to $\bh = 1_T$ and recovers $T$. 

The capacity constraints on $a_e$, however, force $\bh$ to have substantially different support than $T$. By placing the capacity constraints $\bq(e)$ along each $a_e$, the unit of indegree of each node is forced to be spread along various edges. This is where the additional auxiliary edges come in. As a thought experiment, consider the problem of attempting to construct a new arborescence from $T$ where edge $f \in T$ is replaced by some other edge $e \in E$. For $(T \cup \{e\}) \setminus \{f\}$ to remain spanning, $f$ must belong to $P_T(e)$, and each $f \in P_T(e)$ gives a valid swap. The flow network implements these swaps \emph{fractionally}: the amount of flow from $y_e$ to $t_f$ represents a fraction that $e$ substitutes for $f$.

\subsection{Demand-Satisfying Flows are Good}

Now suppose a flow $\bg_0$ in $H_T$ satisfies $\bd$ and let $\bh = \bg_0|_E$. We first aim to obtain a fractional version of \Cref{lem:new-flow-subroutine}: First, $\bh$ clearly satisfies $\bh \leq \bq$ due to the edge capacity constraints, which gives a fractional variant of \Cref{lem:new-flow-subroutine}\ref{cond:marg-pres}. Second, since $\bg_0$ in $H_T$ satisfies $\bd$ with positive demand $n-1$, we have $\bh(E) = n-1$ which gives \Cref{lem:new-flow-subroutine}\ref{cond:edge-num}. Thus, the remaining goal is to show that $\bh$ satisfies \Cref{lem:new-flow-subroutine}\ref{cond:cut-lb}. We show that, by combining $\bh$ with the original tree $T$, one achieves a rooted connectivity of 2 as follows:


\begin{claim}
\label{clm:2conn}
Given a nonnegative flow $\bg_0$ on $H_T$ (even with capacity constraints \ref{cond:marg-pres} dropped) satisfying $\bd$, the corresponding flow $\bh := \bg_0|_E$ in $G$ satisfies $(\mathbf{1}_T + \mathbf{h})(\delta^{-}(S)) \geq 2$ for any non-empty set $S \subseteq V$ with $r \notin S$.
\end{claim}
\begin{proof}
Consider an arbitrary nonempty set $S \subseteq V$ with $r \notin S$. When $T$ does not 1-respect $S$, we have $|\delta^-(S) \cap T| \geq 2$ and thus $(\mathbf{1}_T + \bh)(\delta^-(S)) \geq \mathbf{1}_T(\delta^-(S)) \geq 2$ since $\bh$ is nonnegative. Therefore, below we assume $|\delta^-(S) \cap T| = 1$; we then just need to show that $\bh(\delta^-(S)) \geq 1$. We have:
\[\bh(\delta^-(S)) = \bh(V \times S) - \bh(S \times S).\]
%
%
On the left side of the flow network $H_T$, since $\bg_0$ satisfies demand $\bd$, the flow into each non-root vertex is $1$ under $\bh$. Therefore, $\bh(V \times S) = \sum_{v \in S} \bh(\delta^-(\{v\})) = \sum_{v \in S} 1 = |S|$.

Next, since only one edge in $T$ belongs to $\delta^-(S)$, we know $T[S]$ spans $S$ and thus for any $e \in S \times S$, the path $P_T(e)$ must be a subset of $T[S]$. Therefore, on the right side of the flow network $H_T$, we have $\bh(S \times S) \leq |T[S]| = |S|-1$, as any flow going through $y_e$ where $e \in S \times S$ must flow into some vertex $t_f$ where $f \in P_T(e) \subseteq T[S]$, and each vertex $t_f$ has demand $1$.

Finally, by combining these two terms, we have $\bh(\delta^-(S)) = \bh(V \times S) - \bh(S \times S) \geq |S| - (|S|-1) = 1$, which suffices for our proof.
\end{proof}

\begin{wrapfigure}{R}{0.5\textwidth}
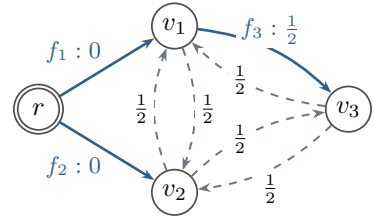

\centering
\includestandalone[mode=tex]{arborescence-counterexample}
\caption{A graph with a demand-satisfying flow whose projection $\mathbf h$ is not a fractional
arborescence. Solid blue arcs form $T=\{f_1,f_2,f_3\}$, and arcs are labelled by $\mathbf h(e)$. Here $\lambda=2$ and $\mathbf q(e)=1/2$, but
$\mathbf h(\delta^-(\{v_1,v_2,v_3\}))=0$. A feasible auxiliary flow routes the two arcs of each pair
$\{v_1,v_2\}$, $\{v_2,v_3\}$, $\{v_1,v_3\}$ to
$t_{f_1}$, $t_{f_2}$, $t_{f_3}$, respectively.}
\label{fig:arborescence_counterexample}
\end{wrapfigure}

We make two notes for the above claim. If we were able to show that $\mathbf{h}(\delta^{-}(S)) \geq 1$ for any non-empty root-excluding set $S$, we could use it to significantly strengthen \Cref{lem:marginal-contraction} and obtain $T'$ where $\EE[1_{T'}\mid T]\le\bq$, instead of  $\EE[\mathbf1_{T'}\mid T]\leq\frac{\mathbf1_T+\bq}{2}$. Unfortunately, this is not true in general.  See \Cref{fig:arborescence_counterexample}.

Second, the above claim only requires the demand $\bd$ is routed, therefore, $(\mathbf{1}_T + \mathbf{h})(\delta^{-}(S)) \geq 2$ holds even when later $\bg_0$ is rounded to another flow $\widehat\bg$ routing $\bd$ that may not satisfy capacity constraints. We exploit this fact.


\subsection{Feasible Solutions Exist}

Next we show the existence of the flow $\bg_0$ routing the demand $\bd$
on $H_T$. This is the only step in which we require $\bq$ to dominate
some point $\bx \in \calA_r$.

\begin{claim}[Feasibility]
\label{clm:feasible}
If $\bq$ dominates some point $\bx \in \calA_r$, then there exists a
non-negative flow $\mathbf{g}_0$ on $H_T$ satisfying $\bd$ with
$\bg_0|_E = \bx$.
\end{claim}

\begin{proof}
It suffices to show that, for any $r$-arborescence $R$, there exists a
non-negative (integral) flow $\bg_R$ on $H_T$ satisfying $\bd$ with
$\bg_R|_E=\mathbf1_R$.
This is because $\bx\in\calA_r$ is a convex combination of the vectors
$\mathbf1_R$: taking the same convex combination of the flows $\bg_R$
gives a flow $\bg_0$ satisfying the same demand $\bd$ with
$\bg_0|_E=\bx\le\bq$, as required.

Fix an $r$-arborescence $R$, and let $H_{R,T}$ be the subgraph of $H_T$
induced by the vertices $y_e$ for $e\in E(R)$ on the left and $t_f$ for
$f\in E(T)$ on the right. In particular, $y_e\to t_f$ is an arc exactly
when $f\in P_T(e)$. We claim that $H_{R,T}$ has a perfect matching.
Such a matching gives the desired flow $\bg_R$: for each matched pair
$(y_e,t_f)$, where $e=(u,v)$, send one unit along
\[
s_v\longrightarrow y_e\longrightarrow t_f.
\]
This satisfies $\bd$ and has projection $\mathbf1_R$, since $R$ has
one edge entering each non-root vertex and the matching uses each
$t_f$ once. See \Cref{fig:feasibility_matching}.

To prove the perfect matching claim, we verify Hall's condition on the right.
For $F\subseteq E(T)$, let $N(F)\subseteq E(R)$ be the in-neighbors
of $F$ in $H_{R,T}$. Deleting $F$ from the underlying undirected tree $T$ leaves
$|F|+1$ connected components. The edges in $N(F)$ are exactly the
edges of $R$ crossing these components, since $P_T(e)$ meets $F$
exactly when the endpoints of $e$ lie in different components.
As $R$ connects these components, $|N(F)|\ge |F|$. Both sides have
$n-1$ vertices, so Hall's theorem gives a perfect matching.
\end{proof}

\begin{figure}[tbp]
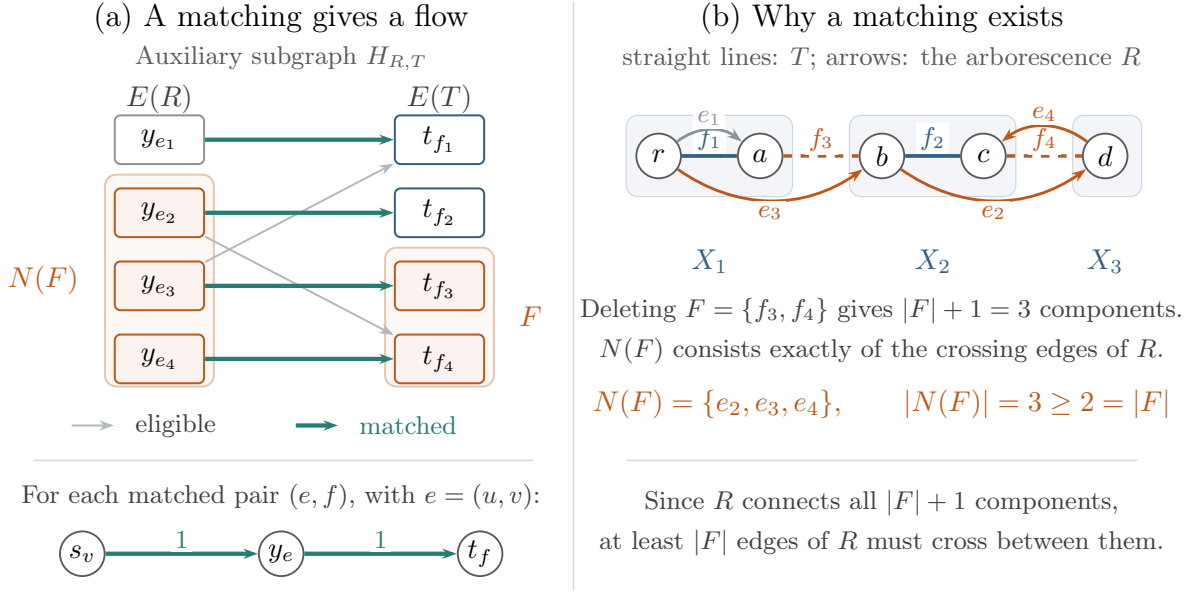

\centering
\includestandalone[mode=tex,width=\textwidth]{feasibility-matching}
\caption{A perfect matching in $H_{R,T}$ gives the desired flow in
\Cref{clm:feasible}. Orange marks $F$ and its in-neighbors $N(F)$,
which connect the components of $T-F$.}
\label{fig:feasibility_matching}
\end{figure}

\subsection{Final Compressed Network Flow Construction and Flow Rounding}
The potentially large part of $H_T$ consists of the auxiliary edges $a_{(e,f)}=(y_e,t_f)$ for $f\in P_T(e)$. To construct the sparser representation $H_T'$, we replace
these edges using a heavy-light decomposition \cite{sleator1981data} of the
underlying tree $T$, rooted at $r$, in a standard way. The construction is as follows:
\begin{enumerate}
    \item Compute a heavy-light decomposition of the tree $T$.
The key property is that each path in $T$ can be decomposed into $O(\log n)$ heavy paths and $O(\log n)$ total intermediate light edges.
\item  For each heavy path in $T$, create a corresponding segment tree \cite{bentley1980optimal} out of Steiner nodes and infinite capacity edges. Steiner nodes are given demand $0$ and the leaves of the segment trees correspond to sinks $t_f$. Segment trees 
have the key property that any interval of tree edges can be represented as a disjoint union of $O(\log n)$ segment-tree nodes. In total, creating the
segment trees for all heavy paths requires $O(n)$ Steiner nodes and $O(n)$ infinite-capacity edges.
\item For each original edge $e \in E$, add infinite-capacity edges from $y_e$ to each node corresponding to a light-edge in the decomposition of $P_T(e)$ into heavy paths, and
add infinite-capacity edges from $y_e$ to the corresponding segment-tree nodes for each heavy path. This amounts to $O(\log^2 n)$ total edges for each $e \in E$.
\end{enumerate}

\begin{figure}
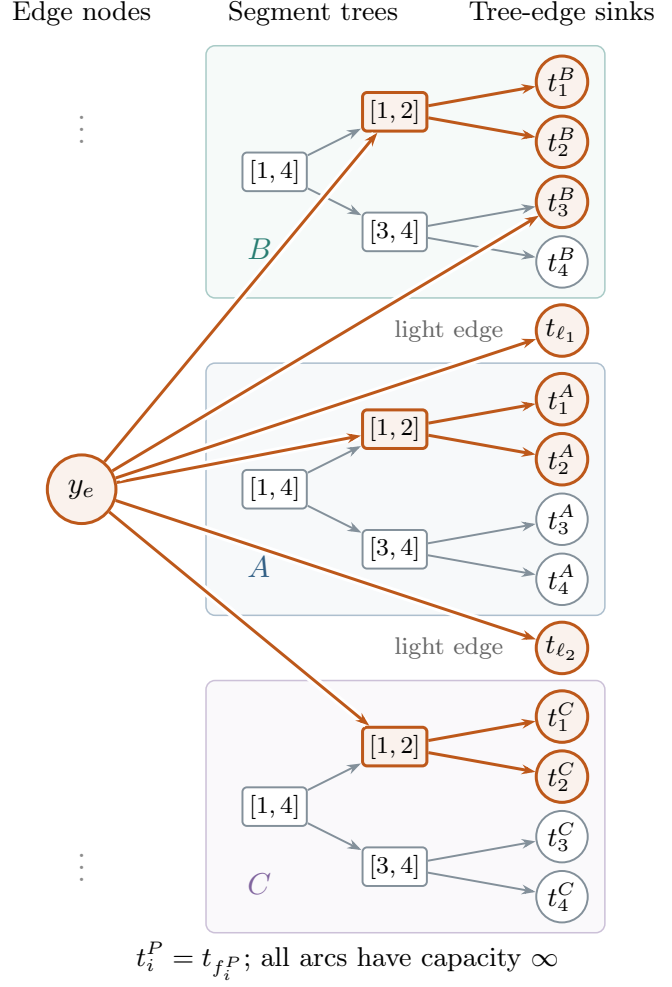

\centering
\includestandalone[mode=tex,width=\textwidth]{efficient-flow-network}
\caption{Efficient representation of the auxiliary edges for $e=(u,v)$.
Colored boxes group sinks by heavy path; orange marks $P_T(e)$ and its
reachable sinks. The node $y_e$ connects to a disjoint segment-tree cover
of each heavy interval and directly to light-edge sinks. For $B[1,3]$,
the cover consists of $[1,2]$ and the singleton sink $t^B_3$.
Internal nodes have demand zero; other branches and edge nodes are omitted.}
\label{fig:efficient_flow_network}
\end{figure}

See \Cref{fig:efficient_flow_network}. In total, the network $H_T'$ has $O(m + n)$ nodes and $O(m \log^2 n)$ edges. Moreover, the heavy-light decomposition and segment trees can be computed in $O(m + n)$ total time. We note that each $y_e$ can reach the same set of sinks $t_f$ in $H_T$ and $H_T'$, both via infinite capacity edges. Therefore, $H_T'$ effectively emulates $H_T$ in the following way: 

\begin{lemma}
   There exists a flow $\mathbf{g}$ in the compressed network $H_T'$ satisfying $\bd$ if and only if there exists a flow $\mathbf{g}_0$ in the network $H_T$ satisfying $\bd$, in each case with $\mathbf{g}|_E = \mathbf{g}_0|_E$. 
\end{lemma}

\begin{proof}
Note that one can decompose any flow in the auxiliary network into weighted paths from nodes $y_e$ to sinks $t_f$. By construction of the auxiliary network, there is a path from a node $y_e$ to sink $t_f$ in $H_T'$ if and only if $f \in P_T(e)$. Therefore, maintaining the flow on each $a_e$ and replacing each weighted path by a single edge (or vice versa) gives the translation between flows $\bg_0$ on $H_T$ and $\bg$ on $H_T'$. 
\end{proof}

Therefore, the compressed network $H_T'$ gives a safe emulation of $H_T$ and the notion of $\bh = \bg_0|_E = \bg|_E$ is consistent. Any $\bh$ coming from $\bg|_E$ can be viewed as coming from $\bg_0|_E$. The compression of $H_T$ into $H_T'$ will justify the total time complexity in \Cref{lem:new-flow-subroutine}. Below we give a full proof.

\begin{proof}[Proof of \Cref{lem:new-flow-subroutine}]
    From \Cref{clm:feasible}, since a feasible flow exists, we can compute such a flow $\mathbf{g}$ in $H_T'$ that routes the demand $\bd$. This uses one max flow call in the compressed network $H_T'$ with $\widetilde{O}(m)$ edges. Let $\bh = \bg|_E$, such a flow $\bh$ satisfies:
    \begin{enumerate}[label = (\alph*)]
        \item $\mathbf{h} \leq \bq$ by capacity constraints,
        \item $\mathbf{h}(E) = n-1$ as the total positive demand is $n-1$,
        \item for all nonempty $S \subseteq V$ with $r \not\in S$, 
    $(\mathbf{1}_T + \mathbf{h})(\delta^-(S)) \geq 2$ by \Cref{clm:2conn}.
    \end{enumerate}

    The only gap here compared to \Cref{lem:new-flow-subroutine} is our current flow $\bh$ is still fractional. To give an integral flow that satisfies the full condition, we apply the following flow rounding framework:

    \begin{theorem}[Flow rounding, \cite{kang2015flow}] \label{thm:flow-rounding}
    Let $H = (V,E, \bc)$ be a flow network and $\bd :V \to \ZZ$ a vertex demand. Let $\textbf{f} : E \to \RR$ be a flow satisfying $\bd$. Then, there exists an $O(|E| \log |V|)$-time randomized algorithm outputting integral flow $\widehat{\textbf{f}}$ also satisfying $\bd$ such that, for each $e \in E$,
    \begin{enumerate}
        \item $\EE[\widehat{\textbf{f}}_e] = \textbf{f}_e$, and 
        \item $\lfloor \textbf{f}_e \rfloor \leq \widehat{\textbf{f}}_e \leq \lceil \textbf{f}_e \rceil$.
    \end{enumerate}
    \end{theorem}
    
    By flow rounding, we get a rounded $\widehat{\bg}$ from $\bg$ which is integral, demand satisfying, 
  and ensures $\mathbb{E}[\widehat{\bg}] = \bg$. By the property of flow rounding, here $\bf{f} = \widehat{\bg}|_E$ will satisfy:
    \begin{enumerate}[label = (\alph*)]
        \item $\mathbb{E}[\mathbf{f}] = \bh \leq \bq$,
        \item $\mathbf{f}(E) = \mathbf{h}(E) = n-1$,
        \item \label{cond:cut-lb2}for all nonempty $S \subseteq V$ with $r \not\in S$, 
    $(\mathbf{1}_T + \mathbf{f})(\delta^-(S)) \geq 2$ by \Cref{clm:2conn}.
    \end{enumerate}
    Note that \ref{cond:cut-lb2} still holds as \Cref{clm:2conn} does not require capacity constraints for $\widehat{\bg}$, or equivalently $\mathbf{f} \leq \bq$. Therefore, the integral flow $\bf{f} = \widehat{\bg}|_E$  satisfies all conditions in \Cref{lem:new-flow-subroutine}.
\end{proof}

%% file: flow-network.tex
\usetikzlibrary{arrows.meta}
\definecolor{treeblue}{HTML}{315F85}
\definecolor{edgegray}{HTML}{6C7782}
\begin{tikzpicture}[
    font=\small,
    >={Stealth[length=2mm,width=1.4mm]},
    vertex/.style={circle,draw=black!70,fill=white,line width=.7pt,
        minimum size=7mm,inner sep=0pt},
    flow vertex/.style={vertex,minimum size=9mm},
    tree edge/.style={->,draw=treeblue,line width=1.15pt},
    flow edge/.style={->,draw=edgegray,line width=.85pt},
    non-tree edge/.style={flow edge,dashed},
    edge label/.style={fill=white,inner sep=2pt,font=\small},
    layer label/.style={align=center,font=\small},
]

\node[font=\normalsize] at (1.6,4.6) {$G$ and its $r$-arborescence $T$};
\node[font=\footnotesize] at (1.6,4.05) {unit capacities; $r$-mincut value $2$};
\node[vertex,double,double distance=1pt] (r) at (1.6,2.65) {$r$};
\node[vertex] (a) at (.1,.4) {$a$};
\node[vertex] (b) at (1.6,-1.8) {$b$};
\node[vertex] (c) at (3.1,.4) {$c$};
\node[vertex] (d) at (1.6,-4.2) {$d$};

\draw[tree edge,<-] (a) -- node[edge label,left,text=treeblue] {$f_1$} (r);
\draw[tree edge,<-] (b) -- node[edge label,above left,text=treeblue] {$f_2$} (c);
\draw[tree edge,<-] (c) -- node[edge label,right,text=treeblue] {$f_3$} (r);
\draw[tree edge,<-] (d) -- node[edge label,right,text=treeblue] {$f_4$} (c);
\draw[non-tree edge,<-] (b) --
    node[edge label,right,text=edgegray] {$e_1$} (a);
\draw[non-tree edge,<-] (d) --
    node[edge label,left,text=edgegray] {$e_2$} (a);
\draw[non-tree edge,<-] (a) to[bend left=25]
    node[edge label,above,text=edgegray] {$e_3$} (c);
\draw[non-tree edge,<-] (c) --
    node[edge label,below,text=edgegray] {$e_4$} (a);

\draw[tree edge] (-.15,-5.7) -- (.45,-5.7);
\node[anchor=west,font=\footnotesize] at (.6,-5.7) {arborescence $T$};
\draw[non-tree edge] (-.15,-6.2) -- (.45,-6.2);
\node[anchor=west,font=\footnotesize] at (.6,-6.2) {non-tree edges};

\draw[black!15,line width=.5pt] (4.45,-6.65) -- (4.45,4.9);
\node[font=\normalsize] at (8.65,4.6) {Corresponding flow network $H_T$};
\node[font=\footnotesize] at (8.65,4.05)
    {$\mathbf q(g)=1/2$; all auxiliary arcs have capacity $\infty$};
\node[layer label] at (5.65,3.25) {Vertices\\$s_v$, $v\ne r$};
\node[layer label] at (8.65,3.25) {Edges\\$y_g$, $g\in E$};
\node[layer label] at (11.65,3.25) {Tree edges\\$t_f$, $f\in E(T)$};

\node[flow vertex] (sb) at (5.65,1.7) {$s_b$};
\node[flow vertex] (sa) at (5.65,-.5) {$s_a$};
\node[flow vertex] (sc) at (5.65,-2.7) {$s_c$};
\node[flow vertex] (sd) at (5.65,-4.9) {$s_d$};

\node[flow vertex] (yf2) at (8.65,2.25) {$y_{f_2}$};
\node[flow vertex] (ye1) at (8.65,1.15) {$y_{e_1}$};
\node[flow vertex] (yf1) at (8.65,.05) {$y_{f_1}$};
\node[flow vertex] (ye3) at (8.65,-1.05) {$y_{e_3}$};
\node[flow vertex] (ye4) at (8.65,-2.15) {$y_{e_4}$};
\node[flow vertex] (yf3) at (8.65,-3.25) {$y_{f_3}$};
\node[flow vertex] (ye2) at (8.65,-4.35) {$y_{e_2}$};
\node[flow vertex] (yf4) at (8.65,-5.45) {$y_{f_4}$};

\node[flow vertex] (tf2) at (11.65,2.25) {$t_{f_2}$};
\node[flow vertex] (tf1) at (11.65,.05) {$t_{f_1}$};
\node[flow vertex] (tf3) at (11.65,-3.25) {$t_{f_3}$};
\node[flow vertex] (tf4) at (11.65,-5.45) {$t_{f_4}$};

\foreach \v/\g in {b/f2,b/e1,a/f1,a/e3,c/e4,c/f3,d/e2,d/f4} {
    \draw[flow edge] (s\v) --
        node[above,inner sep=3pt,font=\footnotesize] {$\tfrac12$} (y\g);
}

\foreach \i in {1,2,3,4} {
    \draw[flow edge] (yf\i) -- (tf\i);
}
\foreach \e/\f in {1/2,1/1,1/3,3/1,3/3,4/1,4/3,2/1,2/3,2/4} {
    \draw[flow edge,preaction={draw=white,line width=3pt,-,
        shorten <=1pt,shorten >=1pt}] (ye\e) -- (tf\f);
}

\node[font=\footnotesize] at (5.65,-6.25) {demand $-1$};
\node[font=\footnotesize] at (8.65,-6.25) {demand $0$};
\node[font=\footnotesize] at (11.65,-6.25) {demand $+1$};

\end{tikzpicture}

%% file: arborescence-counterexample.tex
\usetikzlibrary{arrows.meta}
\definecolor{treeblue}{HTML}{315F85}
\definecolor{edgegray}{HTML}{6C7782}
\begin{tikzpicture}[
    font=\small,
    >={Stealth[length=1.7mm,width=1.2mm]},
    vertex/.style={circle,draw=black!70,fill=white,line width=.65pt,
        minimum size=6mm,inner sep=0pt},
    tree edge/.style={->,draw=treeblue,line width=1pt},
    non-tree edge/.style={->,draw=edgegray,dashed,line width=.75pt},
    edge label/.style={fill=white,inner sep=1.5pt,font=\footnotesize},
]

\node[vertex,double,double distance=.8pt] (r) at (-1.8,0) {$r$};
\node[vertex] (v1) at (0,1.1) {$v_1$};
\node[vertex] (v2) at (0,-1.1) {$v_2$};
\node[vertex] (v3) at (2.3,0) {$v_3$};

\draw[tree edge] (r) --
    node[edge label,above left,text=treeblue] {$f_1:0$} (v1);
\draw[tree edge] (r) --
    node[edge label,below left,text=treeblue] {$f_2:0$} (v2);

\draw[non-tree edge] (v1) to[bend left=20]
    node[edge label,right] {$\tfrac12$} (v2);
\draw[non-tree edge] (v2) to[bend left=20]
    node[edge label,left] {$\tfrac12$} (v1);
\draw[tree edge] (v1) to[bend left=18]
    node[edge label,above,text=treeblue] {$f_3:\tfrac12$} (v3);
\draw[non-tree edge] (v3) to[bend left=18]
    node[edge label,pos=.6] {$\tfrac12$} (v1);
\draw[non-tree edge] (v2) to[bend left=18]
    node[edge label,pos=.4] {$\tfrac12$} (v3);
\draw[non-tree edge] (v3) to[bend left=18]
    node[edge label,below] {$\tfrac12$} (v2);

\end{tikzpicture}

%% file: feasibility-matching.tex
\usetikzlibrary{arrows.meta}
\definecolor{treeblue}{HTML}{315F85}
\definecolor{crossorange}{HTML}{BD581A}
\definecolor{matchgreen}{HTML}{287D72}
\definecolor{edgegray}{HTML}{89949D}

\begin{tikzpicture}[
  font=\small,
  >={Stealth[length=1.8mm,width=1.25mm]},
  vertex/.style={circle,draw=black!65,fill=white,line width=.65pt,
    minimum size=5.5mm,inner sep=0pt},
  tree edge/.style={draw=treeblue,line width=1.3pt},
  deleted edge/.style={draw=crossorange,dashed,line width=1.2pt},
  crossing edge/.style={->,draw=crossorange,line width=1.0pt},
  internal edge/.style={->,draw=edgegray,line width=1.0pt},
  edge label/.style={fill=white,inner sep=1.3pt,font=\footnotesize},
  edge node/.style={rounded corners=2pt,draw=treeblue,fill=white,
    line width=.7pt,minimum width=11mm,minimum height=6mm,inner sep=2pt},
  selected node/.style={edge node,draw=crossorange,fill=crossorange!7},
  eligible/.style={->,draw=edgegray!65,line width=.7pt},
  matched/.style={->,draw=matchgreen,line width=1.5pt},
]

\node[font=\normalsize] at (3.5,5.3)
  {(a) A matching gives a flow};
\node[font=\footnotesize,text=black!65] at (3.5,4.8)
  {Auxiliary subgraph $H_{R,T}$};
\node[text=black!75] at (2.0,4.3) {$E(R)$};
\node[text=black!75] at (5.45,4.3) {$E(T)$};

\filldraw[rounded corners=4pt,fill=crossorange!4,
  draw=crossorange!35,line width=.7pt]
  (1.33,.76) rectangle (2.67,3.37);
\node[anchor=east,text=crossorange] at (1.14,2.07) {$N(F)$};
\filldraw[rounded corners=4pt,fill=crossorange!4,
  draw=crossorange!35,line width=.7pt]
  (4.78,.76) rectangle (6.12,2.47);
\node[anchor=west,text=crossorange] at (6.3,1.62) {$F$};

\node[edge node,draw=edgegray] (ye1) at (2.0,3.8) {$y_{e_1}$};
\node[selected node] (ye2) at (2.0,2.9) {$y_{e_2}$};
\node[selected node] (ye3) at (2.0,2.0) {$y_{e_3}$};
\node[selected node] (ye4) at (2.0,1.1) {$y_{e_4}$};
\node[edge node] (tf1) at (5.45,3.8) {$t_{f_1}$};
\node[edge node] (tf2) at (5.45,2.9) {$t_{f_2}$};
\node[selected node] (tf3) at (5.45,2.0) {$t_{f_3}$};
\node[selected node] (tf4) at (5.45,1.1) {$t_{f_4}$};

\draw[eligible] (ye3) -- (tf1);
\draw[eligible] (ye2) -- (tf4);
\draw[matched] (ye1) -- (tf1);
\draw[matched] (ye3) -- (tf3);
\draw[matched] (ye2) -- (tf2);
\draw[matched] (ye4) -- (tf4);

\draw[eligible] (.92,.28) -- (1.44,.28);
\node[anchor=west,font=\footnotesize,text=black!75] at (1.57,.28)
  {eligible};
\draw[matched] (3.68,.28) -- (4.2,.28);
\node[anchor=west,font=\footnotesize,text=matchgreen] at (4.33,.28)
  {matched};

\draw[black!15,line width=.5pt] (.45,-.15) -- (6.6,-.15);
\node[font=\footnotesize,text=black!75] at (3.5,-.6)
  {For each matched pair $(e,f)$, with $e=(u,v)$:};
\node[vertex] (sv) at (1.05,-1.3) {$s_v$};
\node[vertex] (ye) at (3.5,-1.3) {$y_e$};
\node[vertex] (tf) at (5.95,-1.3) {$t_f$};
\draw[matched] (sv) -- node[edge label,above,text=matchgreen] {$1$} (ye);
\draw[matched] (ye) -- node[edge label,above,text=matchgreen] {$1$} (tf);

\draw[black!15,line width=.5pt] (7.1,-1.75) -- (7.1,5.5);

\node[font=\normalsize] at (10.9,5.3)
  {(b) Why a matching exists};
\node[font=\footnotesize,text=black!65] at (10.9,4.8)
  {straight lines: $T$; arrows: the arborescence $R$};

\filldraw[rounded corners=4pt,fill=treeblue!6,draw=treeblue!25]
  (7.75,3.1) rectangle (9.8,4.1);
\filldraw[rounded corners=4pt,fill=treeblue!6,draw=treeblue!25]
  (10.5,3.1) rectangle (12.55,4.1);
\filldraw[rounded corners=4pt,fill=treeblue!6,draw=treeblue!25]
  (13.25,3.1) rectangle (14.1,4.1);

\node[vertex] (r) at (8.15,3.6) {$r$};
\node[vertex] (a) at (9.4,3.6) {$a$};
\node[vertex] (b) at (10.9,3.6) {$b$};
\node[vertex] (c) at (12.15,3.6) {$c$};
\node[vertex] (d) at (13.65,3.6) {$d$};

\draw[tree edge] (r) -- node[edge label,above,text=treeblue] {$f_1$} (a);
\draw[deleted edge] (a) -- node[edge label,above,text=crossorange] {$f_3$} (b);
\draw[tree edge] (b) -- node[edge label,above,text=treeblue] {$f_2$} (c);
\draw[deleted edge] (c) -- node[edge label,above,text=crossorange] {$f_4$} (d);

\draw[internal edge] (r) to[bend left=38]
  node[edge label,above,text=edgegray] {$e_1$} (a);
\draw[crossing edge] (r) to[bend right=35]
  node[edge label,below,text=crossorange] {$e_3$} (b);
\draw[crossing edge] (b) to[bend right=35]
  node[edge label,below,text=crossorange] {$e_2$} (d);
\draw[crossing edge] (d) to[bend right=38]
  node[edge label,above,text=crossorange] {$e_4$} (c);

\node[text=treeblue] at (8.78,2.3) {$X_1$};
\node[text=treeblue] at (11.53,2.3) {$X_2$};
\node[text=treeblue] at (13.65,2.3) {$X_3$};

\node[font=\footnotesize,align=center,text=black!75] at (10.9,1.45)
  {Deleting $F=\{f_3,f_4\}$ gives $|F|+1=3$ components.\\[3pt]
   $N(F)$ consists exactly of the crossing edges of $R$.};
\node[text=crossorange] at (10.9,.55)
  {$N(F)=\{e_2,e_3,e_4\},\qquad |N(F)|=3\ge 2=|F|$};

\draw[black!15,line width=.5pt] (7.75,-.15) -- (14.05,-.15);
\node[font=\footnotesize,align=center,text=black!75] at (10.9,-.92)
  {Since $R$ connects all $|F|+1$ components,\\[3pt]
   at least $|F|$ edges of $R$ must cross between them.};
\end{tikzpicture}

%% file: efficient-flow-network.tex
\usetikzlibrary{arrows.meta}
\definecolor{heavyA}{HTML}{315F85}
\definecolor{heavyB}{HTML}{287D72}
\definecolor{heavyC}{HTML}{80649E}
\definecolor{selectedpath}{HTML}{BD581A}
\definecolor{auxgray}{HTML}{83909A}

\begin{tikzpicture}[
    font=\small,
    >={Stealth[length=1.65mm,width=1.15mm]},
    tree vertex/.style={circle,draw=black!65,fill=white,
        minimum size=2.5mm,inner sep=0pt,line width=.6pt},
    endpoint/.style={tree vertex,draw=selectedpath,line width=1pt,
        minimum size=6mm},
    heavy/.style={line width=1.35pt},
    light/.style={draw=auxgray,dashed,line width=.8pt},
    path halo/.style={draw=selectedpath!25,line width=6pt,line cap=round},
    edge label/.style={font=\scriptsize,inner sep=1.5pt},
    sink/.style={circle,draw=auxgray,fill=white,line width=.65pt,
        minimum size=6mm,inner sep=0pt,font=\footnotesize},
    reached sink/.style={sink,draw=selectedpath,fill=selectedpath!8,
        line width=1pt},
    interval/.style={rectangle,rounded corners=1.5pt,draw=auxgray,
        fill=white,line width=.65pt,minimum width=7.5mm,
        minimum height=4.5mm,inner sep=1pt,font=\scriptsize},
    chosen interval/.style={interval,draw=selectedpath,
        fill=selectedpath!8,line width=1pt},
    auxiliary/.style={->,draw=auxgray,line width=.65pt},
    reached arc/.style={->,draw=selectedpath,line width=1pt},
    connector/.style={reached arc,
        preaction={draw=white,line width=3pt,-,shorten <=1pt,shorten >=1pt}},
]

\node[font=\normalsize] at (2.1,7.3) {Heavy-light decomposition of $T$};
\node[font=\normalsize] at (9.35,7.3) {Efficient auxiliary network};
\node[font=\footnotesize] at (6.3,6.75) {Edge nodes};
\node[font=\footnotesize] at (9.0,6.75) {Segment trees};
\node[font=\footnotesize] at (11.9,6.75) {Tree-edge sinks};
\draw[black!15,line width=.5pt] (5.3,-4.55) -- (5.3,7.6);

\foreach \i in {0,...,4} {
    \coordinate (A\i) at (2.2,{5.7-1.2*\i});
    \coordinate (B\i) at (.1,{2.8-1.15*\i});
    \coordinate (C\i) at (4.2,{5.1-1.15*\i});
}
\draw[path halo] (B3) -- (B0) -- (A2) -- (A0) -- (C0) -- (C2);

\foreach \x/\y in {2.7/1.0,3.0/1.05,3.3/1.15,3.6/1.25} {
    \draw[black!22,line width=.5pt] (A3) -- (\x,\y);
    \fill[black!25] (\x,\y) circle[radius=.7pt];
}
\node[text=black!35] at (3.1,.62) {$\cdots$};
\node[font=\scriptsize,text=black!45,align=center] at (3.1,.04)
    {omitted\\light edges};

\foreach \P/\col in {A/heavyA,B/heavyB,C/heavyC} {
    \foreach \i [evaluate=\i as \j using int(\i-1)] in {1,...,4} {
        \draw[heavy,draw=\col] (\P\i) --
            node[edge label,left,text=\col] {$f^{\P}_{\i}$} (\P\j);
    }
}
\draw[light] (B0) -- node[edge label,above] {$\ell_1$} (A2);
\draw[light] (C0) -- node[edge label,above] {$\ell_2$} (A0);
\foreach \P in {A,B,C} {
    \foreach \i in {0,...,4} {
        \node[tree vertex] at (\P\i) {};
    }
}
\node[tree vertex,minimum size=6mm,double,double distance=.7pt] at (A0) {$r$};
\node[endpoint] at (B3) {$u$};
\node[endpoint] at (C2) {$v$};
\node[text=heavyA,font=\normalsize] at (2.68,4.95) {$A$};
\node[text=heavyB,font=\normalsize] at (-.35,3.25) {$B$};
\node[text=heavyC,font=\normalsize] at (4.65,5.65) {$C$};

\node at (2.1,-2.5) {$e=(u,v)$};
\draw[path halo] (.5,-3.2) -- (1.1,-3.2);
\node[anchor=west,font=\footnotesize] at (1.25,-3.2) {path $P_T(e)$};
\draw[light] (.5,-3.75) -- (1.1,-3.75);
\node[anchor=west,font=\footnotesize] at (1.25,-3.75) {light edge};

\foreach \P/\yc/\col in {B/4.9/heavyB,A/1.2/heavyA,C/-2.5/heavyC} {
    \path[draw=\col!40,fill=\col!4,rounded corners=3pt,line width=.5pt]
        (7.75,{\yc-1.47}) rectangle (12.4,{\yc+1.47});
    \node[text=\col,anchor=west] at (8.1,{\yc-.9}) {$\P$};
    \node[interval] (\P root) at (8.55,\yc) {$[1,4]$};
    \node[chosen interval] (\P upper) at (9.95,{\yc+.7}) {$[1,2]$};
    \node[interval] (\P lower) at (9.95,{\yc-.7}) {$[3,4]$};
    \foreach \i/\dy in {1/1.05,2/.35,3/-.35,4/-1.05} {
        \node[sink] (\P leaf\i) at (11.9,{\yc+\dy}) {$t^{\P}_{\i}$};
    }
    \foreach \i in {1,2} {
        \node[reached sink] at (\P leaf\i) {$t^{\P}_{\i}$};
    }
    \draw[auxiliary] (\P root) -- (\P upper);
    \draw[auxiliary] (\P root) -- (\P lower);
    \foreach \i in {1,2} {
        \draw[reached arc] (\P upper) -- (\P leaf\i);
    }
    \foreach \i in {3,4} {
        \draw[auxiliary] (\P lower) -- (\P leaf\i);
    }
}
\node[reached sink] at (Bleaf3) {$t^B_3$};

\node[reached sink] (light1) at (11.9,3.05) {$t_{\ell_1}$};
\node[reached sink] (light2) at (11.9,-.65) {$t_{\ell_2}$};
\node[font=\scriptsize,text=black!60,anchor=east] at (11.35,3.05) {light edge};
\node[font=\scriptsize,text=black!60,anchor=east] at (11.35,-.65) {light edge};

\node[reached sink,minimum size=8mm,font=\small] (ye) at (6.3,1.2) {$y_e$};
\node[text=black!45] at (6.3,5.5) {$\vdots$};
\node[text=black!45] at (6.3,-3.1) {$\vdots$};

\foreach \target in {Bupper,Bleaf3,Aupper,Cupper,light1,light2} {
    \draw[connector] (ye) -- (\target);
}
\node[font=\footnotesize] at (9.4,-4.3)
    {$t_i^P=t_{f_i^P}$; all arcs have capacity $\infty$};

\end{tikzpicture}